\documentclass{article}
\usepackage[utf8]{inputenc}

\usepackage{graphicx}
\usepackage{enumerate,url}
\usepackage{mathrsfs}
\usepackage{amsfonts,amsmath,amssymb}
\usepackage{fullpage}

\usepackage{bm}
\usepackage[standard,thmmarks,thref]{ntheorem}
\theoremsymbol{\ensuremath{\square}}
\usepackage{enumitem,color}
\usepackage[usenames,dvipsnames,svgnames,table]{xcolor}
\usepackage[normalem]{ulem}

\usepackage{arydshln} 
\usepackage[section]{algorithm}
\usepackage{algorithmic}
\usepackage[obeyFinal]{todonotes}

\usepackage{chicago}
\usepackage[colorlinks=true,citecolor=blue]{hyperref}

\newcommand{\tuple}[1]{\langle{#1}\rangle}
\newcommand{\EL}{\ensuremath{\mathscr{E\!L}}}
\newcommand{\ELPP}{\ensuremath{\EL^{++}}}
\newcommand{\ALC}{\ensuremath{\mathscr{A\!L\!C}}}

\newcommand{\dland}{\ensuremath{\sqcap}}
\newcommand{\dlto}{\ensuremath{\sqsubseteq}}
\newcommand{\sfl}{\sffamily\slshape}

\renewcommand{\phi}{\varphi}
\renewcommand{\emptyset}{\varnothing}

\definecolor{darkblue}{rgb}{0,0,.7}
\definecolor{purple}{rgb}{.65,0,.65}
\definecolor{mygray}{RGB}{150,150,150}

\title{Extending Description Logic EL++ with Linear Constraints on the Probability of Axioms%
	\thanks{This study was financed in part by the Coordena\c{c}\~{a}o de Aperfei\c{c}oamento de Pessoal de N\'{i}vel Superior -- Brasil (CAPES) -- Finance Code 001.}
}

\author{Marcelo Finger\thanks{Partly supported by Fapesp projects 2015/21880-4 and 2014/12236-1 and CNPq grant PQ 303609/2018-4.}
		\\Department of Computer Science\\
		University of S\~{a}o Paulo, Brazil.\\
		\url{mfinger@ime.usp.br}
}

\date{}

\begin{document}
\maketitle

\begin{abstract}
	One of the main reasons to employ a description logic such as EL or EL++ is the fact that it has efficient, polynomial-time algorithmic properties such as deciding consistency and inferring subsumption. However, simply by adding negation of concepts to it, we obtain the expressivity of description logics whose decision procedure is {ExpTime}-complete. Similar complexity explosion occurs if we add probability assignments on concepts. To lower the resulting complexity, we instead concentrate on assigning probabilities to Axioms (GCIs). We show that the consistency detection problem for such a probabilistic description logic is NP-complete, and present a linear algebraic deterministic algorithm to solve it, using the column generation technique. We also examine and provide algorithms for the probabilistic extension problem, which consists of inferring the minimum and maximum probabilities for a new axiom, given a consistent probabilistic knowledge base. 
	
	An earlier version of this work has appeared as~\cite{Fin2019b1}.  Here we detail the column generation method and present a detailed example. 
	
\end{abstract}

\section{Introduction}
\label{sec:intro}

The logic \ELPP\ is one of the most expressive description logics in which the complexity of inferential reasoning  is tractable~\cite{BBL2005}.   A direct consequence of this expressivity is that, by adding extra features to this language, its complexity easily grows exponentially.  By inferential complexity we mean the complexity of decision problems such as consistency  detection, finding a model that satisfies a set of constraints, or Axiom subsumption. All such problems are tractable in \ELPP.

In this work we are interested in adding probabilistic reasoning capabilities to \ELPP;  however, depending on how those reasoning capabilities are added to the language, the inferential complexity can explode beyond exponential time. As shown in Section~\ref{sec:pconcept}, by extending \ELPP\ with probabilistic constraints over concepts, inferential reasoning becomes \textsc{ExpTime}-hard.  Such an approach was employed in many times in the literature, either by  enhancing expressive description logics such as \ALC~\cite{Hei1994,Luk2008,G-BJLS2011,JGLS2011},  or by adding probabilistic capabilities to the family of \EL-like logics~\cite{LS2010,BJLS2017}.

In this work, we study a different way of extending description logics with probabilistic reasoning capabilities, namely by applying probabilities to GCI Axioms.  One of our goals is to reduce the complexity of probabilistic reasoning in description logics. Another goal is to deal with the modelling situation in which a GCI Axiom is not always true, but one can assign (subjectively) a probability to  its validity. Consider the following example describing one such situation.  

\begin{example}\label{ex:dengue0}\rm	
	Consider the following medical situation, in which a patient may have symptoms which are caused buy a disease. However, some diseases cause only very nonspecific symptoms,  such as high fever, skin rash and joint pain, which may also be caused by several other diseases.
	Dengue is one such desease with mostly nonspecific symptoms.  Dengue is a mosquito-borne viral disease and more than half of the world population lives at risk of contracting it.  Among its symptoms are high fever, joint pains and skin eruptions (rash).  These symptoms are common but not all patients present all symptoms.  Such an uncertain  situation allows for probabilistic modelling. 
	
	In a certain hospital, joint pains are caused by dengue in 20\% of the cases; in the remaining 80\% of the cases, there is a patient whose symptoms include joint pains whose cause is \emph{not} attributable to dengue.   Also, a patient having high fever has some probability having dengue, which increases 5\% if the patient also has a rash. If those probabilistic constraints are satisfiable, one can also ask the minimum and maximum probability that a given patient is a suspect of suffering from dengue.
\end{example}

By adding probability constraints to axioms, we hope to model such a situation.  Furthermore we will show that the inferential complexity in this case remains ``only’’ NP-complete.   In fact, our approach extends some previous results which considered adding probabilistic capabilities only to ABox statements~\cite{FWC2011}.  By using \ELPP\ as the underlying formalism, ABox statements can be formulated as a particular case of GCI Axioms, so the approach here has that of~\cite{FWC2011} as a particular case, but with inferential reasoning remaining in the same complexity class.

The rest of the paper proceeds as follows.  Section~\ref{sec:prelim} presents the formal \ELPP-framework and Section~\ref{sec:pgci} introduces probabilities over axioms, and define the probabilistic satisfiability and  probabilistic extension problems.  Section~\ref{sec:linalg} presents an algorithm for probabilistic satisfiability that combines \ELPP-solving with linear algebraic methods, such as column generation.  Finally, Section~\ref{sec:ext} presents an algorithm for the probabilistic extension problem, and then we present our conclusions in Section~\ref{sec:conc}.

\section{Preliminaries}
\label{sec:prelim}

We concentrate on the description language \ELPP\ without concrete domains~\cite{BBL2005}. We start with a signature consisting of a triple of countable sets $\mathsf{N}=\tuple{\mathsf{N_C,N_R,N_I}}$ where $\mathsf{N_C}$ is a set of \emph{concept names}, $\mathsf{N_R}$ is a set of \emph{role names} and $\mathsf{N_I}$ is a set of \emph{individual names}.  The basic \emph{concept description} are recursively defined as follows:
\begin{itemize}
	\item $\top$, $\bot$ and concept names in $\mathsf{N_C}$ are (simple) concept descriptions;
	\item if $C,D$ are concept descriptions, $C \dland D$ is a (conjunctive) concept description;
	\item if $C$ is a concept description and $r \in \mathsf{N_R}$, $\exists r.C$ is an (existential) concept description;
	\item if $a \in \mathsf{N_I}$, $\{a\}$ is a (nominal) concept description;
\end{itemize}

If $C, D$ are concept descriptions an \emph{axiom}, also called a \emph{general concept inclusion} (GCI), is an expression of the form $C \dlto D$. If $r, r_1, \ldots, r_k \in \mathsf{N_R}$ then $r_1 \circ \cdots \circ r_k \dlto r$ is a \emph{role inclusion} (RI). A finite set of axioms is called a \emph{TBox} and a finite set of axioms and RIs is called a \emph{constraint box} (CBox).

A \emph{concept assertion} is an expression of the form $C(a)$, where $a \in \mathsf{N_I}$ and $C$ is a concept description; a \emph{role assertion} is an expression of the form $r(a,b)$, where $a,b \in \mathsf{N_I}$ and $r \in \mathsf{N_R}$.  A finite set of concept and role assertions forms an assertion box (\emph{ABox}).

Semantically, we consider an \emph{interpretation} $\mathcal{I} = \tuple{\Delta^\mathcal{I}, \cdot^\mathcal{I}}$. The domain $\Delta^\mathcal{I}$ is a non-empty set of individuals and the interpretation function $\cdot^\mathcal{I}$ maps each concept name $A \in \mathsf{N_C}$ to a subset $A^\mathcal{I} \subseteq \Delta^\mathcal{I}$, each role name $r \in \mathsf{N_R}$ to a binary relation $r^\mathcal{I} \subseteq \Delta^\mathcal{I} \times \Delta^\mathcal{I}$ and each individual name $a \in \mathsf{N_I}$ to an individual $a^\mathcal{I} \in \Delta^\mathcal{I}$. The extension of $\cdot^\mathcal{I}$ to arbitrary concept descriptions is inductively defined as follows.

\begin{itemize}
	\item $\top^\mathcal{I} = \Delta^\mathcal{I}$, $\bot^\mathcal{I} = \emptyset$;
	\item $(C \dland D)^\mathcal{I} = C^\mathcal{I} \cap D^\mathcal{I}$;
	\item $(\exists r.C)^\mathcal{I} = \{x \in \Delta^\mathcal{I} | \exists y \in C^\mathcal{I}, \tuple{x,y} \in r^\mathcal{I} \}$;
	\item $(\{a\})^\mathcal{I} = \{a^\mathcal{I}\}$.
\end{itemize}

The interpretation $\mathcal{I}$ \emph{satisfies} an axiom $C \dlto D$ if $C^\mathcal{I} \subseteq D^\mathcal{I}$ (represented as $\mathcal{I} \models C \dlto D$); the RI  $r_1 \circ \cdots \circ r_k \dlto r$ is satisfied by  $\mathcal{I}$ (represented as $\mathcal{I} \models r_1 \circ \cdots \circ r_k \dlto r$) if  $r_1^\mathcal{I} \circ \cdots \circ r_k^\mathcal{I} \subseteq r^\mathcal{I}$. A model $\mathcal{I}$ satisfies the assertion  $C(a)$ (represented as $\mathcal{I} \models C(a)$) if $a^\mathcal{I} \in C^\mathcal{I}$ and satisfies the assertion $r(a,b)$ (represented as $\mathcal{I} \models r(a,b)$) if $\tuple{a^\mathcal{I},b^\mathcal{I}} \in r^\mathcal{I}$.
Given a CBox $\mathcal{C}$, we write $\mathcal{I} \models \mathcal{C}$ if $\mathcal{I} \models C \dlto D$ for every axiom $C \dlto D \in \mathcal{C}$ and $\mathcal{I} \models r_1 \circ \cdots \circ r_k \dlto r$ for every role inclusion in $\mathcal{C}$ . Similarly, given an ABox $\mathcal{A}$, we write $\mathcal{I} \models \mathcal{A}$ if $\mathcal{I}$ satisfies all its assertions.

Given a CBox $\mathcal{C}$, we say that it  \emph{logically entails} an axiom $C \dlto D$, represented as $\mathcal{C} \models C \dlto D$, if for every interpretation $\mathcal{I} \models \mathcal{C}$ we have that $\mathcal{I} \models  C \dlto D$.

Note that in \ELPP\ there is no need for an explicit ABox, for we have that  $\mathcal{I} \models  C(a)$ iff $\mathcal{I} \models  \{a\} \dlto C$; and $\mathcal{I} \models  r(a,b)$ iff $\mathcal{I} \models  \{a\} \dlto \exists r.\{b\}$.

Given a CBox, one of the important problems for \ELPP\ is to determine its \textit{consistency}, namely the existence of a common model which jointly validates all expressions in the CBox. There is a polynomial algorithm which decides \ELPP-consistency~\cite{BBL2005a}. 

This decision process can be used to provide a PTIME \emph{classification} of an \EL\ CBox. Given a CBox $\mathcal{C}$, the set $ \mathsf{BC}_\mathcal{C}$ of \emph{basic concepts descriptions for $\mathcal{C}$} is given by 

\[ \mathsf{BC}_\mathcal{C} = \{\top,\bot\} \cup \left\{\frac{}{} C \in \mathsf{N_C} | C  \textrm{ used in } \mathcal{C}\right\} \cup \left\{\frac{}{}\{a_i\} | a_i \in \mathsf{N_I} \textrm{ used in } \mathcal{C}\right\} . \]


\begin{example}\label{ex:dengue1}\rm
	Consider a CBox representing the situation described in Example~\ref{ex:dengue0}; this modelling is adapted from~\cite{FWC2011}.
	
	\smallskip
	\begin{tabular}{l|l}
		\begin{minipage}[t]{0.4\linewidth}
			\noindent
			The following TBox $\mathcal{T}_0$ describes basic knowledge on deseases:
			
			\smallskip
			
			{\scriptsize\sfl
				High-fever $\sqsubseteq$ Symptom
				
				Joint-pain $\sqsubseteq$ Symptom
				
				Rash $\dlto$ Symptom
				
				Dengue $\sqsubseteq$ Disease
				
				Symptom $\sqsubseteq$ $\exists$hasCause.Disease
				
				Patient $\sqsubseteq$ $\exists$suspectOf.Disease
				
				Patient $\sqsubseteq$ $\exists$hasSymptom.Symptom
				
				$\exists$hasSymptom.($\exists$hasCause.Dengue)$\sqsubseteq$ \\
				\hspace*{9em}$\exists$suspectOf.Dengue
				
			}  
		\end{minipage}
		~
		&
		~
		\begin{minipage}[t]{0.5\linewidth}
			\noindent
			And the following ABox presents John's symptoms.
			
			\smallskip
			
			{\scriptsize\sfl\hspace*{-1em}
				\begin{tabular}{ll@{$]$}}
				Patient(john) & $[\equiv\{$john$\} \dlto$ Patient\\
				
				High-fever($s_1$) & $[\equiv\{s_1\} \dlto$ High-fever\\
				
				hasSymptom(john, $s_1$) & $[\equiv\{$john$\} \dlto \exists$hasSymptom.$\{s_1\}$\\
				
				
				
				Joint-pain($s_2$) & $[\equiv\{s_2\} \dlto$ Joint-pain\\

				hasSymptom(john, $s_2$) & $[\equiv\{$john$\} \dlto \exists$hasSymptom.$\{s_2\}$\\
				\end{tabular}
			}
		\end{minipage}
	\end{tabular}
	\smallskip

	\noindent
	Note that the uncertain information on dengue and its symptoms is not represented by the CBox above.
	\end{example}

\section{Extending \ELPP\ with Probabilistic Constraints}
\label{sec:pgci}

One of the main reasons to employ a description logic such as \ELPP\  is the fact that it has polynomial-time algorithmic properties such as deciding and inferring subsumption.    However, it is well known that  simply by adding negation of concepts to \ELPP,  we obtain the  expressivity of  description logic \ALC\,  whose decision procedure  is \textsc{ExpTime}-complete~\cite{BHLS2017}.  This complexity blow up can also be expected when adding probabilistic constraints. 

\subsection{Why Not Assign Probability to Concepts?}
\label{sec:pconcept}

When we are dealing with probabilistic constraints on description logic, one of the first ideas is to apply conditional or unconditional probability constraints to concepts.   In fact,  such an approach was employed in several enhancements of description logics with   probabilistic reasoning capabilities,  e.g. as~\cite{Hei1994,Luk2008,LS2010,BJLS2017}.

However, one can see how such an approach would lead to problems if applied to \ELPP. For each  concept $C$   one can  define  an associated concept $\bar{C}$  subject to the following constraints:

\begin{align*}
P(C) + P(\bar{C}) &= 1\\
P(C \dland \bar{C}) &= 0
\end{align*}

Without going into the (non-trivial) semantic details of concept probabilities,  it is intuitively clear that those statements force $\bar{C}$  to be the negation of $C$.  In fact,  the first statement expresses that $C$  and $\bar{C}$  are complementary  and the second statement  expresses that they are  disjoint;  together they mean that  interpretation of $C$  and $\bar{C}$  form a partition of the domain, and thus $\bar{C}$ is the negation of $C$.  As a consequence, the expressivity provided by probabilities over concepts  adds to \ELPP\ the expressivity  of \ALC, and as a consequence the complexity of deciding axiom  subsumption becomes \textsc{ExpTime}-hard.   Detailed complexity analysis can be found in~\cite{BJLS2017}. 

To lower the resulting complexity, we refrain from  assigning probabilities to concepts and instead concentrate on assigning probabilities to axioms.

\subsection{Probability Constraints over Axioms}
\label{sec:probax}

Assume there is a finite number of interpretations, $\mathcal{I}_1, \ldots, \mathcal{I}_m$; let $P$  be a mapping that attributes to each $\mathcal{I}_i$ a positive value $P(\mathcal{I}_i) \geq 0$ such that  $\sum_{i=1}^{m} P(\mathcal{I}_i) = 1$. 

Then given an axiom $C \dlto D$, its probability is given by:
\begin{align}
	P(C \dlto D) = \sum_{\mathcal{I}_i \models C \dlto D} P(\mathcal{I}_i)~~ .	
\end{align}
Note that this definition contemplates the probability of ABox elements; for example the probability  $P(C(a)) =P(\{a\} \dlto C)$.

Given axioms $C_1 \dlto D_1, \ldots, C_\ell \dlto D_\ell$ and rational numbers $b_1, \ldots, b_\ell;q$, a \emph{probabilistic constraint} consist of  the linear combination:
\begin{align} \label{eq:ctrnt}
	b_1 \cdot P(C_1 \dlto D_1) + \cdots b_\ell \cdot P(C_\ell \dlto D_\ell) \bowtie q ~~,
\end{align}
where $\bowtie\ \in \{\leq, \geq,=\}$. A \emph{PBox} is a set of probabilistic constraints.  A \emph{probabilistic knowledge base} is a pair $\tuple{\mathcal{C,P}}$, where $\mathcal{C}$ is a CBox and $\mathcal{P}$ a PBox. Note that the axioms occurring in the PBox need not occur in the CBox, and in general they do not occur in it.

The intuition behind the probability of a GCI can perhaps be better understood if  seen by its complement.    So the probability of an axiom $C \dlto D$  is $p$ if the probability of its failure is $1-p$,  that is,  the probability of finding a model $\mathcal{I}$ in which there exists an  individual $a$  that is in  concept $C$  but not in concept $D$, $\mathcal{I} \models C(a)$ and  $\mathcal{I} \not\models D(a)$. Under this point of view, $P(C \dlto D)=p$ if there is a probability $p$ of finding a model in which either no individual instantiates concept $C$ or all individual instances of concept $C$ are  also individual instances of concept $D$. This has as a consequence the following, somewhat unintuitive behavior:  if $C$ is a ``rare’’ concept in the sense that most models have no instances of $C$,  then the probability $P(C \dlto D )$ tends to be quite high for any $D$,  for it  has as lower bound the probability of a model not having any instances of $C$.

Note that this  intuitive view also covers ABox statements, which can be expressed  as axioms   of the form $\{a\} \dlto C$ and $\{a\} \dlto \exists r.\{b\}$. But in these cases, all models always satisfy  the nominal $\{a \}$, so e.g. $P(\{a\} \dlto C) = p$  simply means that the probability of finding a model  in which $a$  is an instance of concept $C$ is $p$.

\subsection{Probabilistic Satisfaction and Extension Problems}
\label{sec:satext}

A probabilistic knowledge base $\tuple{\mathcal{C,P}}$ is satisfied by interpretations $\mathcal{I}_1, \ldots, \mathcal{I}_m$  if there exists a probability distribution $P$ over the interpretations such that
\begin{itemize}
	\item if $P(\mathcal{I}_i) > 0$ then $\mathcal{I}_i \models \mathcal{C}$;
	\item all probabilistic constraints in $\mathcal{P}$ hold.
\end{itemize}
This means that an interpretation can have a positive probability mass only if it satisfies CBox $\mathcal{C}$,  and the composition of all those interpretations  must verify the probability of constraints in $\mathcal{P}$. A knowledge base is \emph{satisfiable} if there exists a set of interpretations and a probability distribution over them that satisfy it.

\begin{definition}\rmfamily
	The \emph{probabilistic satisfiability problem}  for the logic \ELPP\  consists of,  given a probabilistic knowledge base $\tuple{\mathcal{C,P}}$,  decide if it is  satisfiable. 
\end{definition}

\begin{definition}
	The \emph{probabilistic extension problem}  for the logic \ELPP\  consists of,  given a satisfiable probabilistic knowledge base $\tuple{\mathcal{C,P}}$ and an axiom $C \dlto D$,  find the minimum and maximum values of $P(C \dlto D)$ that are satisfiable with $\tuple{\mathcal{C,P}}$.
\end{definition}

\begin{example}\label{ex:dengue2}\rm
	We create a probabilistic knowledge base by extending the CBox presented in Example~\ref{ex:dengue1} with the uncertain information described in Example~\ref{ex:dengue0}.
	
	Dengue symptoms are nonspecific, so in some cases the high fever is actually caused by dengue, represented by {\sfl  Ax1 := High-fever $\dlto \exists$hasCause.Dengue}, and in some other cases we may have a combination of high fever and rash being caused by dengue, represented by {\sfl  Ax2 := High-fever $\dland$ Rash $\dlto \exists$hasCause.Dengue}.  And the fact that joint pains are caused by dengue is represented by {\sfl  Ax3 := Joint-pain $\dlto \exists$hasCause.Dengue}. None of the axioms {\sfl Ax1, Ax2} or {\sfl Ax3} is always the case, but there is a probability that dengue is, in fact, the cause.	
	The following probabilistic statements represents uncertain knowledge on the relationship between dengue and its symptoms, as observed in a hospital.  

	\smallskip
	{\scriptsize\sfl
		\begin{tabular}{lp{30em}}
			$P($Ax2$) -  P($Ax1$) = 0.05$ & \rmfamily The probability of dengue being the cause is 5\% higher when both high fever and rash are symptoms, over just having high fever;\\
			$P($Ax3$) = 0.2$ & \rmfamily 20\% of cases of joint pain are caused by dengue.
		\end{tabular}
	}
	\smallskip
	
	\noindent
	We want to know if this probabilistic database is consistent and, in case it is, we want to find upper and lower bounds for the probability  that John is a suspect of having dengue, {\sfl $p_{lb} \leq P(\exists$suspectOf.Dengue(john)) $ \leq p_{ub}$}.	
\end{example}

In order to provide algorithms that  tackle both the decision and the extension problems,  we provide a linear algebra formulation of those problems.

\subsection{A Linear Algebraic View of Probabilistic Satisfaction and Extension Problems}

Initially, let us consider only restricted probabilistic constraints of the form $P(C_i \dlto D_i) = p_i$.  Consider a restricted probabilistic knowledge base $\tuple{\mathcal{C,P}}$ in which the number of probabilistic constraints is $|\mathcal{P}| = k$. Let $p$ be a vector of size $k$ of probabilistic constraint values.  Consider a finite number of interpretations, $\mathcal{I}_1, \ldots, \mathcal{I}_m$, and let us build a $k \times m$ matrix $A$ of $\{0,1\}$ elements $a_{ij}$ such that 
\[a_{ij} = 1 \textit{ iff } \mathcal{I}_j \models C_i \dlto D_i ~~\]

Note that column $A^j$ contains the evaluations by interpretation $\mathcal{I}_j$ of the axioms submitted to probabilistic constraints. Given a CBox $\mathcal{C}$ and sequence of $n$ axioms $C_1 \dlto D_1, \ldots, C_n \dlto D_n$, a $\{0,1\}$-vector $u$ of size $n$ \emph{represents} a $\mathcal{C}$-satisfiable interpretation $\mathcal{I}$ if $\mathcal{I} \models \mathcal{C}$, and  $c_i = 1$ iff $\mathcal{I} \models C_i \dlto D_i
$ for $1 \leq i \leq n$. The idea is to assign positive probability mass $pi_j >0$ only if  $A^j$ represents a $\mathcal{C}$-satisfiable interpretation.

Let $\pi$ be  a vector of size $m$ representing a probability distribution. Consider the following set of constraints associated to $\tuple{\mathcal{C,P}}$,  expressing  the fact that $\pi$ is a probability distribution that respects the constraints given by matrix $A$:

\begin{align}
	A \cdot \pi &= p \nonumber \\
	\sum_{j=1}^m \pi_j &= 1 \label{eq:pel1} \\
	\pi &\geq 0 \nonumber 	
\end{align}
 
The fact that constraints~\eqref{eq:pel1} actually represent satisfiability is given by the following.

\begin{lemma}\label{lemma:pel1}
	A probabilistic knowledge base $\tuple{\mathcal{C,P}}$ with restricted  probabilistic constraints is satisfiable iff there is a vector $\pi$ that satisfies its associated constraints ~\eqref{eq:pel1}.
\end{lemma}

When the probabilistic knowledge base is satisfiable, the number $m$ of interpretations associated to the columns of matrix $A$ may be exponentially large with respect to the number $k$ of constraints in $\mathcal{P}$.   However, Carath\'eodory's Theorem~\cite{Eck93} guarantees that if there is a solution to~\eqref{eq:pel1} then there is also a small solution, namely one with at most $k+1$  positive values.

\begin{lemma}\label{lemma:carat1}
	If constraints~\eqref{eq:pel1} have a solution then there exists a solution $\pi$ with at most $k+1$ values such that $\pi_j > 0$.
\end{lemma}

Now instead of considering only a restricted form of probability constraints,  let us consider  constraints of the form ~\eqref{eq:ctrnt} as defined in Section~\ref{sec:pgci},  namely
\begin{align} 
b_{i1} \cdot P(C_1 \dlto D_1) + \cdots + b_{i\ell} \cdot P(C_\ell \dlto D_\ell) \bowtie q_i ~~, \nonumber
\end{align}
where $b_{ij}, q_i \in \mathbb{Q}$,  $\bowtie\ \in \{\leq, \geq,=\}$ and $i=1, \ldots k$. 

We assume there are at most $\ell$ axioms  mentioned  in $\mathcal{P}$,  such that $b_{i,j}=0$  if $P(C_j \dlto D_j)$ does not occur at  constraint $i$.  Consider a matrix $B_{k\times  \ell}$ and a vector $x$ of size  $\ell$.  We now have the following  set of associated constraints to the probabilistic knowledge base $\tuple{\mathcal{C,P}}$, extending~\eqref{eq:pel1}:

\begin{align}
	B \cdot x &= q \nonumber \\
	A \cdot \pi &= x \label{eq:pel++} \\
	\sum_{j=1}^m \pi_j &= 1 \nonumber \\
	x,\pi &\geq 0 \nonumber 	
\end{align}

As before, $A$'s columnns are $\{0,1\}$-representations of the validity of the axioms occurring in $\mathcal{P}$ under the interpretation $\mathcal{I}_j$. Constraints~\eqref{eq:pel++} are \emph{solvable} if there are vectors $x$ and $\pi$ that verify all conditions. Analogously, the solvability of constraints~\eqref{eq:pel++} characterize  the satisfiability of probabilistic knowledge bases with unrestricted constraints.

\begin{lemma}\label{lemma:pel++}
	A probabilistic knowledge base $\tuple{\mathcal{C,P}}$ is satisfiable if and only if its associated set of constraints ~\eqref{eq:pel++} are solvable.
\end{lemma}

\begin{example}\label{ex:dengue3}\rm
	Consider four interpretations for the knowledge base described in Example~\ref{ex:dengue2}. Interpretation $\mathcal{I}_1$  satisfies CBox $\mathcal{C}$ of Example~\ref{ex:dengue1} and also axioms {\sfl Ax1, Ax2, Ax3}. Interpretation $\mathcal{I}_2$  satisfies $\mathcal{C}$ and axioms {\sfl Ax2, Ax3} but not {\sfl Ax1}. Interpretation $\mathcal{I}_3$  satisfies $\mathcal{C}$  and only axiom {\sfl Ax3}. Interpretation $\mathcal{I}_4$ satisfies only $\mathcal{C}$ but none of the axioms. We then consider a probability distribution $\pi$, such that $\pi(\mathcal{I}_1) = 5\%$, $\pi(\mathcal{I}_2) = 5\%$, $\pi(\mathcal{I}_3) = 10\%$, $\pi(\mathcal{I}_4) = 80\%$.  The following shows that all probabislistic restrictions are satisfied.
	
	\[
	\begin{array}{l}
	\sf Ax1 \\ \sf Ax2 \\ \sf Ax3 \\ 1
	\end{array}
	~~~
	\left[
	\begin{array}{cccc}
	1 & 0 & 0 & 0 \\
	1 & 1 & 0 & 0 \\
	1 & 1 & 1 & 0 \\
	1 & 1 & 1 & 1
	\end{array}
	\right]
	\cdot
	\left[
	\begin{array}{l}
	0.05 \\ 0.05 \\ 0.10 \\ 0.80
	\end{array}
	\right]
	=
	\left[
	\begin{array}{c}
	0.05 \\ 0.10 \\ 0.20 \\ 1.00
	\end{array}
	\right]	
	\]
	
	\noindent
	So 	$P(${\sfl Ax2}$) - P(${\sfl Ax1}$) = 0.05$ and $P(${\sfl Ax3}$) = 0.2$.
\end{example}

When constraints~\eqref{eq:pel++} are \emph{solvable}, vector $x$ has size $\ell = O(k)$,  but vector $\pi$ can be  exponentially large in $k$.   
By a simple linear algebraic trick, constraints of the form~\eqref{eq:pel++} can  he presented in the following form:
\begin{align}
	C \cdot \pi^x &= d \label{eq:pel} \\
	\pi^x &\geq 0 \nonumber 	
\end{align}
\noindent
In fact,  it suffices  to make:
\[
C = \left[
	\begin{array}{c;{3pt/3pt}c}
	0 & B \\ \hdashline[3pt/3pt]
	A  & -I_\ell \\ \hdashline[3pt/3pt]
	\mathbf{1} & 0
	\end{array}
	\right]
;~~~~~~~~
d = \left[
	\begin{array}{c}
	q \\ \hdashline[3pt/3pt]
	0 \\ \hdashline[3pt/3pt]
	1 
	\end{array}
	\right]
;~~~~~~~~
\pi^x = \left[\begin{array}{c} \pi \\ \hdashline[2pt/2pt] x \end{array}\right]
\]
where $I_\ell$ is the identity matrix, and $\mathbf{1}$ is a row of $|\pi|$ 1's.   When we say that the column $C^j$ represents a $\mathcal{C}$-satisfiable interpretation, we actually mean that the part of $C^j$ that corresponds to some column $A^j$ that represents a $\mathcal{C}$-satisfiable interpretation, its $k$-initial positions are $0$ and its last element is $1$. Note that $C$ has $k + \ell + 1$ rows and $|\pi|+\ell$ columns. Again,  Carath\'eodory's Theorem guarantees small solutions.

\begin{lemma}\label{lemma:carat2}
	If constraints~\eqref{eq:pel++} have a solution then there exists a solution $\pi^x$ with at most $k+\ell+1$ values such that $\pi^x_j > 0$.
\end{lemma}

We now show that probabilistic satisfiability is NP-hard.

\begin{lemma}\label{lemma:nphard}
	The satisfiability problem for probabilistic knowledge bases is NP-hard.
\end{lemma}

\begin{proof}
	We reduce SAT to probabilistic satisfiability over \ELPP;   unlike PSAT\footnote{PSAT, or Probabilistic SATisfiability, consists of determining the satisfiability of a set of probabilistic assertions on classical propositional formulas~\cite{FDB2011,DF2015b,DCF2014}.},  it does not suffice to set all probabilities to 1,  as \ELPP\  is  decidable in polynomial time.  Instead,  we show how to represent 3-SAT clauses (i.e.  disjunction of three  literals)  as a set of probabilistic axioms, basically probabilistic ABox statements.   For that,  consider a set of propositional variables $x_1, \ldots, x_n$ upon which the set $\Gamma$ of clauses of  the SAT problem  are built.   On the probabilistic  knowledge base side,  consider a single individual $a$ and $2n$ basic concepts $ X_1, \ldots, X_n$ and $\overline{X}_1, \ldots, \overline{X}_n$,   subject to the following $2n$  restrictions:	
	\begin{align} \label{eq:neg}
		P(a \dlto X_i) + P(a \dlto \overline{X}_i) &= 1 \\
		P(a \dlto X_i \dland \overline{X}_i) &= 0 \nonumber
	\end{align}
	
	The idea is to represent the propositional atomic information $x_i$  by the axiom $a \dlto X_i$,  its  negation by $a \dlto \overline{X}_i$, and the  fact that  a clause $y_i \lor \ldots \lor y_m$  holds is represented by the probabilistic statement 
	\begin{align} \label{eq:clause}
		P(a \dlto \overline{Y}_i \dland \ldots \dland \overline{Y}_m) = 0 .
	\end{align}
	
	Given $\Gamma$, we build a probabilistic knowledge base $\tuple{\varnothing,\mathcal{P}}$ by the representation~\eqref{eq:clause} of the clauses in $\Gamma$ plus $2n$  assertions of the form~\eqref{eq:neg}. We claim that $\Gamma$  is satisfiable iff $\tuple{\varnothing,\mathcal{P}}$ is.	In fact, suppose $\Gamma$ is satisfiable by valuation $v$, make  a \ELPP\   model $\mathcal{I}$ such that $\mathcal{I} \models a \dlto X_i$ iff $v(x_i) = 1$   and assign probability 1 to $\mathcal{I}$;   clearly $\tuple{\varnothing,\mathcal{P}}$ is satisfiable. Now suppose $\tuple{\varnothing,\mathcal{P}}$ is satisfiable, so there exists an \ELPP\ model $\mathcal{I}$  which is assigned  probability   strictly bigger than 0.  Construct a valuation $v$  such that $v(x_i) = 1$ iff $\mathcal{I} \ \models a \dlto X_i$. Clearly $v(\Gamma)=1$, otherwise there is a clause $y_i \lor \ldots \lor y_m$ in $\Gamma$ such that $v(y_i \lor \ldots \lor y_m)=0$ and thus $\mathcal{I} \models a \dlto \overline{Y}_i$ for $i=1, \ldots, m$; then $P(a\dlto \overline{Y}_i \dland \ldots \dland \overline{Y}_m) \geq P(\mathcal{I}) > 0$, contradicting~\eqref{eq:clause}.	
\end{proof}

\begin{theorem}\label{th:pel++}
	The satisfiability problem for probabilistic knowledge bases is NP-complete.%
\end{theorem}

\begin{proof}
	Lemma~\ref{lemma:carat2} provides a small witness for every problem, such that by guessing that witness we can show in  polynomial time that the constraints  are solvable; so the problem is in NP.  Lemma~\ref{lemma:nphard} provides NP-hardness.
\end{proof}

\section{Column Generation Algorithm for Probabilistic Knowledge Base Satisfiability}
\label{sec:linalg}

An algorithm for deciding probabilistic knowledge base satisfiability has to provide a means to find a solution for restrictions~\eqref{eq:pel++} if one exists; otherwise determine no solution is possible. Furthermore, we will assume that the constraints are presented in format~\eqref{eq:pel1}.


We now provide a method similar to PSAT-solving to decide the satisfiability of probabilistic knowledge base $\tuple{\mathcal{C,P}}$. We construct a vector $c$ of costs whose size is the same as size of $\pi^x$ such that $c_j \in \{0,1\}$, $c_j=1$ if column $C^j$ satisfies the following condition: either the first $k$ positions are not 0, or the next $\ell$ cells representing $A^j$ correspond to an interpretation that \emph{does not} satisfy the CBox $\mathcal{C}$, or the last position of $C^j$ is not $1$; if $C^j$ is one of the last $\ell$ columns, or its first $k$ elements are 0 and the next $\ell$ elements are a representation of an interpretation $A^j$ that is $\mathcal{C}$-satisfiable and its last element is $1$, 
then $c_j=0$. Then we generate the following optimization problem associated to \eqref{eq:pel1}.

\begin{align} \label{eq:PELPP_prog}
\begin{array}{lll}
\min 				& c' \cdot \pi^x \\
\mbox{subject to} 	& C \cdot \pi^x=d \\
& \pi^x \geq 0
\end{array}
\end{align}

\begin{lemma}\label{thm:PELPP_prog}
	Given a probabilistic knowledge base $\tuple{\mathcal{C,P}}$ and its associated linear algebraic restrictions \eqref{eq:pel++}, $\tuple{\mathcal{C,P}}$ is satisfiable if, and only if, minimization problem \eqref{eq:PELPP_prog} has a minimum such that $c'\pi=0$.
\end{lemma}

Condition $c'\pi=0$ means that only the columns of ${A}^j$ corresponding to $\mathcal{C}$-satisfiable interpretations can be attributed probability $\pi_j>0$, which immediately leads to solution of \eqref{eq:PELPP_prog}. Minimization problem \eqref{eq:PELPP_prog} can be solved by an adaptation of the simplex method with column generation such that the columns of $C$ corresponding to columns of $A$ are generated on the fly. The simplex method is a stepwise method which at each step considers a basis consisting of $k+\ell+1$ columns of matrix $C$ and computes its associated cost~\cite{BT1997}. The processing proceeds by finding a column of $C$ outside the basis, creating a new basis by substituting one of the basis columns by this new column such that the associated cost never increases. To guarantee the cost never increases, the new column $C^j$ to be inserted in the basis has to obey a restriction called reduced cost given by $\tilde{c}_j = c_j-c_{B_a} {B_a}^{-1} C^j \leq 0$, where $c_j$ is the cost of column $C^j$, ${B_a}$ is the basis and $c_{B_a}$ is the cost associated to the basis. Note that in our case, we are only inserting columns that represent $\mathcal{C}$-satisfiable interpretations, so that we only insert columns of matrix $C$ and their associated cost $c_j = 0$. Therefore, every new column $C^j$ to be inserted in the basis has to obey the inequality
\begin{align}\label{eq:PELPP_cost}
c_{B_a} {B_a}^{-1} C^j \geq 0.
\end{align}
Note that the first $k$ positions in $C^j$ are 0 and the last one is always 1. 

A column $C^j$ representing a $\mathcal{C}$-satisfying interpretation may or may not satisfy condition \eqref{eq:PELPP_cost}. We call an interpretation that does satisfy \eqref{eq:PELPP_cost} as \emph{cost reducing interpretation}. Our strategy for column generation is given by finding cost reducing interpretations for a given basis.

\begin{lemma}\label{thm:interp_cost}
	There exists an algorithm that decides the existence of cost reducing interpretations whose complexity is in NP.
\end{lemma}

\begin{proof}
	Since we are dealing with a CBox in \ELPP, the existence of satisfying interpretations is polynomial-time and thus in NP, we can guess one such equilibrium and in polynomial time both verify it is a $\mathcal{C}$-satisfying interpretation and that is satisfies~\eqref{eq:PELPP_cost}. \qed
\end{proof}

We can actually build a deterministic algorithm for Lemma \ref{thm:interp_cost} by reducing it to a SAT problem. In fact, computing \ELPP\ satisfiability can be encoded in a 3-SAT formula $\phi$; the condition \eqref{eq:PELPP_cost} can also be encoded by a 3-SAT formula $\psi$ in linear time, e.g. by Warners algorithm~\cite{War1998}, such that the SAT problem consisting of deciding $\phi\cup\psi$ is satisfiable if, and only if, there exists a cost reducing interpretation. Furthermore its valuation provides the desired column $C^j$, after prefixing it with $k$ 0's and appending a 1 at its end. This SAT-based algorithm we call the \emph{\ELPP-Column Generation Method}.  In practice, column generation tries \emph{first} to output one of the last $\ell$ columns in $C$; if the insertion of one such column causes $det(B_a)=0$ or $\pi^x \not\geq 0$, or if all the last $\ell$ $C$-columns are in the basis, the proper\ELPP-Column Generation Method is invoked.

\begin{algorithm}
	\caption{PKBSAT-CG: a probabilistic knowledge base solver via Column Generation\label{alg:PCEviaCG}}
	\textbf{Input:} A probabilistic knowledge base $\tuple{\mathcal{C,P}}$ and its associated set of restrictions in format~\eqref{eq:pel1}.
	
	\textbf{Output:} No, if $\tuple{\mathcal{C,P}}$ is unsatisfiable. Or a solution
	$\tuple{B_a,\pi^x}$ that minimizes \eqref{eq:PELPP_prog}.
	
	\begin{algorithmic}[1]
		\STATE $B_a^{(0)} := {I}_{k+\ell+1};$ \label{lin:iniini}
		\STATE $s := 0$, ${\pi^x}^{(s)} = (B_a^{(0)})^{-1} \cdot d$ and $c^{(s)} = [1 \cdots 1]';$ \label{lin:iniend}
		\WHILE{$c^{(s)}{}' \cdot {\pi^x}^{(s)} \neq 0$} 
		\label{lin:loop}
		\STATE $y^{(s)} = \mathit{GenerateColumn}(B_a^{(s)},\mathcal{C},c^{(s)});$ \label{lin:cond}
		\IF{Column generation failed} 
		\RETURN No;~~ \label{lin:fail}  \COMMENT{probabilistic knowledge base is
			unsatisfiable}
		\ELSE
		\STATE $B_a^{(s+1)} = \mathit{merge}(B_a^{(s)}, y^{(s)});$ \label{lin:merge}
		\STATE $s\!\!+\!\!+$, recompute ${\pi^x}^{(s)} := (B_a^{(s-1)})^{-1} \cdot d$; $c^{(s)}$ the costs of $B_a^{(s)}$ columns;
		\ENDIF
		\ENDWHILE\label{lin:endloop}
		\RETURN $\tuple{B_a^{(s)},{\pi^x}^{(s)}}$;~~  \COMMENT{probabilistic knowledge base is satisfiable} \label{lin:end}
	\end{algorithmic}
\end{algorithm}

Algorithm \ref{alg:PCEviaCG} presents the top level probabilistic knowledge base decision procedure. Lines \ref{lin:iniini}--\ref{lin:iniend} present the initialization of the algorithm. We assume the vector $p$ is in descending order. At the initial step we make $B^{(0)} = U_{K+1}$, this forces $\pi^{(0)}_{K+1} = p_{K+1} \geq 0$, $ \pi^{(0)}_{j} = p_{j} -p_{j+1} \geq 0, 1 \leq j \leq K$; and $c^{(0)} = [c_1 \cdots c_{K+1}]'$, where $c_j=0$ if column $j$ in $B^{(0)}$ is an interpretation; otherwise $c_j=1$. Thus the initial state $s=0$ is a feasible solution.

Algorithm \ref{alg:PCEviaCG} main loop covers lines \ref{lin:loop}--\ref{lin:endloop} which contains the column generation strategy at beginning of the loop (line \ref{lin:cond}). If column generation fails the process ends with failure in line \ref{lin:fail}; the correctness of unsatisfiability by failure is guaranteed by Lemma~\ref{thm:PELPP_prog}. Otherwise a column is removed and the generated column is inserted in a process we called \textit{merge} at line \ref{liWe have thus proved the following result.
	n:merge}. The loop ends successfully when the objective function (total cost) $c^{(s)}{}' \cdot {\pi^x}^{(s)}$ reaches zero and the algorithm outputs a probability distribution $\pi^x$ and the set of interpretations columns in $B_a$, at line \ref{lin:end}. 

The procedure \textit{merge} is part of the simplex method which guarantees that given a column $y$ and a feasible solution $\tuple{B_a,\pi^x}$ there always exists a column $j$ in $B_a$ such that if  $B_a[j:=y]$ is obtained from $B_a$ by replacing column $j$ with $y$, then there is $\tilde{\pi^x}\geq 0$ such that $\tuple{B_a[j:=y],\tilde{\pi^x}}$ is a feasible solution.

\subsection{Column Generation Procedure}
\label{sec:colgenproc}

Column generation is based on the cost reduction condition \eqref{eq:PELPP_cost}, which we repeat here:
\begin{align}\label{eq:PELPP_cost2}
c_{B_a} {B_a}^{-1} C^j \geq 0.
\end{align}

Recall that matrix $C$ is of the form 
\[
C = \left[
\begin{array}{c;{3pt/3pt}c}
0 & B \\ \hdashline[3pt/3pt]
A  & -I_\ell \\ \hdashline[3pt/3pt]
\mathbf{1} & 0
\end{array}
\right]
\]

So, column generation first tries to insert a cost decreasing column from the last $\ell$ columns in $C$; this involves verifying if condition~\eqref{eq:PELPP_cost2} holds for any of the $\ell$ rightmost columns, which are known from the start and do not need to be generated.  If one such column is found, it is returned.

If no such column is found, however, \ELPP-Column Generation Method described next is invoked.  As the number of columns of matrix $A$ is potentially exponentially large and thus not stored.  Note that the first $k$ positions in a generated column $C^j$ are all 0 and the last entry is always 1; the remaining $\ell$ positions are a column of matrix $A$ representing an \ELPP-interpretation $\mathcal{I}$; those positions are 0-1 values, where 1 represents $\mathcal{I} \models C_i \dlto D_i$ and 0 representing the existence of some domain element $b$ such that $\mathcal{I} \models C_i(b)$ but $\mathcal{I} \not\models \dlto D_i(b)$, $1 \leq i \leq \ell$.  Thus the elements of a generated $c^j$ are all 0-1, and we identify them with valuations of a satisfying assignment of a SAT formula $\Gamma$ obtained as follows:
\begin{enumerate}
	\item $\Gamma_1$ is obtained by translating the inequality~\eqref{eq:PELPP_cost2} into a set of clauses;  this can be done, for instance, using the procedure described by~\cite{War1998}.
	\item $\Gamma_2$ is a rendering of the \ELPP-decision procedure as a SAT formula for the \ELPP-satisfiability bt some interpretation $\mathcal{I}$ of the given set of axioms on which linear conditions are imposed, $C_1 \dlto D_1, \ldots, C_\ell \dlto D_\ell$.
\end{enumerate}
Formulas $\Gamma_1$ and $\Gamma_2$ share variables indicating whether $\mathcal{I} \models C_i \dlto D_i$, $1 \leq i \leq \ell$.  We take $\Gamma = \Gamma_1 \cup \Gamma_2$, and send it to a SAT solver.  If $\Gamma$ is satisfiable, we obtain from the satisfying valuation a column that is cosat reducing, due to n$\Gamma_1$ and that represents an \ELPP-model, due to $\Gamma_2$.

As the constraints of the sumplex method are thus respected, and it is an always terminating procedure, we have the following result.

\begin{theorem}
	Algorithm \ref{alg:PCEviaCG} decides probabilistic knowledge base satisfiability using column generation.
\end{theorem}

A detailed example is provided illustrated the procedure details.

\begin{example}\rm
	We now show a step-by-step solution of the satisfiability of the dengue example using Algorithm~\ref{alg:PCEviaCG} and column generation procedure as above.  At each step $s$ we are going to show the basis $B_a^{(s)}$, the basis cost vector $c^{(s)}$, the partial solution $\pi^{(s)}$, the current \textit{cost} $=c^{(s)}{}' \cdot \pi^{(s)}$ and the generated column $y$.
	
	The columns generated correspond to \ELPP-models that have to satisfy the restrictions 
	\begin{align}\label{model}
		\mathrm{Ax}_1 \models_{\ELPP} \mathrm{Ax}_2
	\end{align}
	
	Each row of the basis corresponds to some restriction.  Initially, the basis is the identity matrix, the basis cost vector is all 1's, idicating that all columns do not correspond to any model satisfying~\eqref{model}. 

	\begin{align*}
		\color{gray} c^{(0)}{}'=~ &
		\left[\begin{array}{cccccc}
			1 & 1 & 1 & 1 & 1 & 1
		\end{array}
		\right]&
		{\color{gray} \textit{cost}=1.25}\\
		\color{gray}
		\begin{array}{rr}
			B_a^{(0)}=&P_1 \\ &P_2 \\ &Ax_1 \\ &Ax_2 \\ &Ax_3 \\ &1
		\end{array}
		&
		\left[
		\begin{array}{cccccc}
			1 & 0 & 0 & 0 & 0 & 0 \\
			0 & 1 & 0 & 0 & 0 & 0 \\
			0 & 0 & 1 & 0 & 0 & 0 \\
			0 & 0 & 0 & 1 & 0 & 0 \\
			0 & 0 & 0 & 0 & 1 & 0 \\
			0 & 0 & 0 & 0 & 0 & 1
		\end{array}
		\right]
		\cdot
		\left[
		\begin{array}{c}
			0.05 \\ 0.20 \\ 0 \\ 0 \\ 0 \\ 1
		\end{array}
		\right]
		\hspace*{-4em}
		&
		=
		\left[
		\begin{array}{c}
			0.05 \\ 0.20 \\ 0 \\ 0 \\ 0 \\ 1
		\end{array}
		\right]		
	\end{align*}

	As described above, column generation first tries to insert a cost decreasing column from the last $\ell$ columns in $C$, which are known a priory.  In our case we have the following $B$-equalities and the corresponding columns:
	\begin{align*}
			P(Ax_2) - P(Ax_1) &= 0.05\\
			P(Ax_3) &= 0.20
	\end{align*}
	\[
		C = \left[
		\begin{array}{c;{3pt/3pt}c}
			0 & B \\ \hdashline[3pt/3pt]
			A  & -I_\ell \\ \hdashline[3pt/3pt]
			\mathbf{1} & 0
		\end{array}
		\right]
		~~~~
		\left[
		\begin{array}{c}
			B \\ \hdashline[3pt/3pt]
			-I_\ell \\ \hdashline[3pt/3pt]
			0
		\end{array}
		\right]
		=
		\left[ 
			{
			\begin{array}{r}
			-1 \\ 0 \\ -1 \\ 0 \\ 0 \\ 0
			\end{array}}
			{
			\begin{array}{r}
			1 \\ 0 \\ 0 \\ -1 \\ 0 \\ 0
			\end{array}}
			{
			\begin{array}{r}
			0 \\ 1 \\ 0 \\ 0 \\ -1 \\ 0
			\end{array}}
		\right] 
	\]
	where each corresponds to axioms $\mathrm{Ax}_1$, $\mathrm{Ax}_2$ and  $\mathrm{Ax}_3$, respectively.  It occurs that those columns satisfy the column reduction inequality~\eqref{eq:PELPP_cost2}, and are inserted in the basis in the order $\mathrm{Ax}_3$, $\mathrm{Ax}_2$, $\mathrm{Ax}_3$; also note that the rightmost column does correspond to a model satisfying restriction~\eqref{model}, so after 4 column generation steps we have the following state:

	\begin{align*}
		\color{gray} c^{(4)}{}'=~ &
		\left[\begin{array}{rrrrrr}
			~~~0 & ~~0 & ~~0 & 1 & 1 & 0
		\end{array}
		\right]&
		{\color{gray} \textit{cost}=1.25}\\
		\color{gray}
		\begin{array}{rr}
		B_a^{(4)}=&P_1 \\ &P_2 \\ &Ax_1 \\ &Ax_2 \\ &Ax_3 \\ &1
		\end{array}
		&
		\left[
		\begin{array}{rrrrrr}
		1 & 0 & -1 & 0 & 0 & 0 \\
		0 & 1 & 0 & 0 & 0 & 0 \\
		0 & 0 & -1 & 0 & 0 & 0 \\
		-1 & 0 & 0 & 1 & 0 & 0 \\
		0 & -1 & 0 & 0 & 1 & 0 \\
		0 & 0 & 0 & 0 & 0 & 1
		\end{array}
		\right]
		\cdot
		\left[
		\begin{array}{c}
		0.05 \\ 0.20 \\ 0 \\ 0.05 \\ 0.20 \\ 1
		\end{array}
		\right]
		\hspace*{-4em}
		&
		=
		\left[
		\begin{array}{c}
		0.05 \\ 0.20 \\ 0 \\ 0 \\ 0 \\ 1
		\end{array}
		\right]		
	\end{align*}
	
	Note that the inserted columns now correspont to positions of basis cost vector with value 0.  The choice of which columns leave the basis is performed by the \textit{merge} procedure, which is a linear algebraic method that ensures that $\pi \geq 0$.  Note that total cost has not decreased so far, which is always a possibility as condition~\eqref{eq:PELPP_cost2} only ensures that the coat is non-increasing.  As all the rightmost $B$-columns have already been inserted in the basis, we have to proceed to a proper column generation process in which restriction~\eqref{model} needs to be respected as well as the following inequality:
	\begin{align*}
		c'_{B_a} B_a^{-1} C_j = 
		\left[1~1~-1~1~1~0\right]\cdot
		\left[
		0~ 0 ~ Ax_1 ~ Ax_2 ~ Ax_3 ~ 1
		\right]'&= -Ax_1 +Ax_2+Ax_2\geq 0
	\end{align*}
	We transform the inequality above to a SAT formula, together with a transformation of restriction~\eqref{model} into another SAT formula, and submit to a SAT solver that generates a satisfying valuation indicating that there is an \ELPP-model that satisfies axioms 2 and 3 but not axiom 1, thus generating the column $[0~0~0~1~1~1 ]'$ which the \textit{merge} procedures inserts as the fourth column, thus generating the state:
	\begin{align*}
		\color{gray} c^{(5)}{}'=~ &
		\left[\begin{array}{rrrrrr}
		~~~0 & ~~0 & ~~0 & ~0 & 1 & 0
		\end{array}
		\right]&
		{\color{gray} \textit{cost}=0.15}\\
		\color{gray}
		\begin{array}{rr}
		B_a^{(5)}=&P_1 \\ &P_2 \\ &Ax_1 \\ &Ax_2 \\ &Ax_3 \\ &1
		\end{array}
		&
		\left[
		\begin{array}{rrrrrr}
		1 & 0 & -1 & 0 & 0 & 0 \\
		 0 & 1 & 0 & 0 & 0 & 0 \\
		 0 & 0 & -1 & 0 & 0 & 0 \\
		-1 & 0 & 0 & 1 & 0 & 0 \\
		0 & -1 & 0 & 1 & 1 & 0 \\
		0 &  0 & 0 & 1 & 0 & 1
		\end{array}
		\right]
		\cdot
		\left[
		\begin{array}{c}
		0.05 \\ 0.20 \\ 0 \\ 0.05 \\ 0.15 \\ 0.95
		\end{array}
		\right]
		\hspace*{-4em}
		&
		=
		\left[
		\begin{array}{c}
		0.05 \\ 0.20 \\ 0 \\ 0 \\ 0 \\ 1
		\end{array}
		\right]		
	\end{align*}
	Note that the total cost has decreased for the first time.  The \textit{merge} process chooses a column to leave the basis so as to guarantee that the partial solution $\pi \geq 0$, but it does not ensure that the leaving column is one with non-zero cost.  In fact, it is a coincidence that in this example all columns that left the basis had non-zero cost; on he other hand, it is \emph{by construction} that all entering columns have zero cost.
	
	Finally, we proceed with column generation.  As before, we obtain the inequality
	
	\[ 	c'_{B_a} B_a^{-1} C_j = 
		Ax_1 -Ax_2+Ax_2\geq 0	\\
	\]
	which together with restriction~\eqref{eq:PELPP_cost2} allows for a model in which the three axioms in focus are all true; as before, such a state is obtained by submitting a SAT-encoded formula to a SAT solver.  We obtain the sixth step in the column generation process:
	\begin{align*}
		\color{gray} c^{(6)}{}'=~ &
		\left[\begin{array}{rrrrrr}
		~~~0 & ~~0 & ~~0 & ~0 & 0 & 0
		\end{array}
		\right]&
		{\color{gray}\bf \textit{\bfseries cost}=0}\\
		\color{gray}
		\begin{array}{rr}
		B_a^{(6)}=&P_1 \\ &P_2 \\ &Ax_1 \\ &Ax_2 \\ &Ax_3 \\ &1
		\end{array}
		&
		\left[
		\begin{array}{rrrrrr}
		1 & 0 & -1 & 0 & 0 & 0 \\
		 0 & 1 & 0 & 0 & 0 & 0 \\
		 0 & 0 & -1 & 0 & 1 & 0 \\
		-1 & 0 & 0 & 1 & 1 & 0 \\
		0 & -1 & 0 & 1 & 1 & 0 \\
		0 &  0 & 0 & 1 & 1 & 1
		\end{array}
		\right]
		\cdot
		\left[
		\begin{array}{c}
		0.20 \\ 0.20 \\ 0.15 \\ 0.05 \\ 0.15 \\ 0.80
		\end{array}
		\right]
		\hspace*{-4em}
		&
		=
		\left[
		\begin{array}{c}
		0.05 \\ 0.20 \\ 0 \\ 0 \\ 0 \\ 1
		\end{array}
		\right]		
	\end{align*}
	As the total cost has reached 0, we know the problem is satisfiable.  The last three columns, whose last element is one, correspond to three \ELPP-models on which a probability distribution was obtained, given by the corresponding elements of $\pi^{(6)}$, $0.05, 0.15, 0.80$.  The initial three columns correspond to the $B$-columns and, in the order presented correpond to axioms 2, 3 and 1 and from $\pi^{(6)}$ we can read their probabilities: 0.20, 0.20 and 0.15; note that the initial equations are all respected and the example is finished.
\end{example}

\section{Algorithm for the Probabilistic Extension Problem}
\label{sec:ext}

We now analyse the problem of probabilistic knowledge base extension. Given a satisfiable knowledge base, our aim is to find the maximum and minimum probabilistic constraints for some axiom $C \dlto D$ maintaining satisfiability. Given a precision $\varepsilon=2^{-k}$, the algorithm works by making a binary search through the binary representation of the possible constraints to $C \dlto D$,  solving a probabilistic knowledge base satisfiability problem in each step.

Algorithm \ref{alg:ex} presents a procedure to solve the maximum extension problem. We invoke $\textrm{PKBSAT-CG}(\tuple{\mathcal{C,P}})$ several times in the process. Obtaining the minimum extension is easily adaptable from Algorithm \ref{alg:ex}.

\begin{algorithm}
	\caption{PKBEx-BS: a solver for probabilistic knowledge base extension via Binary Search\label{alg:ex}}
	\textbf{Input:} A satisfiable probabilistic knowledge base $\tuple{\mathcal{C,P}}$, an axiom $C \dlto D$, and a precision $\varepsilon>0$.
	
	\textbf{Output:} Maximum $P(C \dlto D)$ value with precision $\varepsilon$.
	
	\begin{algorithmic}[1]
		\STATE $k := \lceil|\log\varepsilon|\rceil$;
		\STATE $j := 1$, $v_{min} := 0$, $v_{max} := 1$;
		\IF{$\textit{PKBSAT-CG}(\mathcal{C,P}\cup\{P(C \dlto D)=1\}) = \mathrm{Yes}$}
		\STATE $v_{min} := 1$;
		\ELSE
		\WHILE{$j\leq k$}
		\STATE{$v_{max} = v_{min} + \frac{1}{2^j}$};
		\IF{$\textit{PKBSAT-CG}(\mathcal{C,P}\cup\{P(C \dlto D)\geq v_{max}\}) = \mathrm{Yes}$}
		\STATE $v_{min} := v_{max}$;
		\ENDIF
		\STATE $j\!\!+\!\!+$;
		\ENDWHILE
		\ENDIF
		\RETURN $v_{min}$;
	\end{algorithmic}
\end{algorithm}

Suppose the goal is to find the maximum possible value for constraining $C \dlto D$. Iteration 1 solves PKBSAT for $P(C \dlto D)=1$; if it is satisfiable, $\overline{P}(C \dlto D)=1$, else $\overline{P}(C \dlto D)=0$ with precision $2^0$=1, and it can be refined by solving PKBSAT for $P(C \dlto D)=0.5$; if it is satisfiable, $\overline{P}(C \dlto D)=0.5$, else $\overline{P}(C \dlto D)=0$, both cases with precision $2^{-1}=0.5$. One more iteration gives precision $2^{-2}=0.25$, and it consists of solving PKBSAT for $P(C \dlto D)=0.75$ in case the former iteration was satisfiable, otherwise $P(C \dlto D)=0.25$. The proceeds until the desired precision is reached, which takes $|\log 2^{-k}|+1= k+1$ iterations.

\begin{theorem}\label{thm:pkbex}
	Given a precision $\varepsilon>0$, probabilistic knowledge base extension can be obtained with $O(|\log\varepsilon|)$ iterations of  probabilistic knowledge base satisfiability.
\end{theorem}

\begin{example}\rm
	If we continue he previous examples, by applying Algorithm~\ref{alg:ex}, we obtain that 
	\begin{center}
		\sfl $0.20 \leq P(\exists$suspectOf.Dengue(john)) $ \leq 0.95$.
	\end{center}
	that is, the probability of John having 
	Dengue lies between twenty percent and ninety five percent. Such a high spread means that knowing lower and upper bounds for probability is not really informative. 
\end{example}

It is important to note that this binary search is not the only way to solve the extension problem.  A modification of the column generation procedure is also possible, in which a distinct optimization objective function is used, and in which only models satisfying {\sfl$\exists$suspectOf.Dengue(john)} are generated, could also be used.  We omit the details here.

\section{Conclusions and Further Work}
\label{sec:conc}

In this paper we have extended the logic \ELPP\ with probabilistic reasoning capabilities over GCI axioms, without causing an exponentially-hard complexity blow up in reasoning tasks. We have provided deterministic algorithms based on logic and linear algebra for the problems of probabilistic satisfiability and probabilistic extension, and we have demonstrated that the decision problems are NP-complete.

In the future, we plan to explore more informative probabilistic measures, such as probabilities under minimum entropy distributions and the dealing of conditional probabilities, instead of only focusing on probabilities of $\dlto$-axioms, as was done here. We also plan to study fragments of the logics presented here in the search for tractable fragments of probabilistic description logics.

\bibliographystyle{chicago}
\bibliography{mf}
	
\end{document}